\documentclass[runninghead]{llncs} 

\usepackage[a4paper,hmargin=1.25in,vmargin=1.25in,footskip=40pt]{geometry}

\spnewtheorem{theo}{Theorem}{\bfseries}{\itshape}
\spnewtheorem{lemm}[theo]{Lemma}{\bfseries}{\itshape}
\spnewtheorem{cor}[theo]{Corollary}{\bfseries}{\itshape}
\spnewtheorem{defi}[theo]{Definition}{\bfseries}{\itshape}
\spnewtheorem{obs}[theo]{Observation}{\bfseries}{\itshape}

\usepackage{cite}
\usepackage[colorinlistoftodos,prependcaption,textsize=tiny]{todonotes}
\usepackage{graphicx}
\usepackage[font=footnotesize,labelfont=bf]{caption}

\usepackage{tikz}
\usetikzlibrary{arrows,arrows.meta,math,patterns,decorations,snakes}
\tikzstyle{server} = [rectangle, draw = black, fill = white]
\tikzstyle{request} = [fill=black,circle,scale=0.5]
\tikzstyle{augment} = [thick]
\tikzstyle{alternative} = [thick, dashed, ->]
\tikzstyle{matching} = [thick, draw = black]
\tikzstyle{online} = [very thick, draw = utblue]
\tikzset{
	schraffiert/.style={pattern=north west lines,pattern color=#1},
	schraffiert/.default=black
}
\tikzset{snake it/.style={decorate, decoration={snake, segment length=2mm, amplitude = .5 mm}}}
\usepackage{color,xcolor}
\definecolor{utgreen}{RGB}{90,90,90}
\definecolor{SIjp1}{RGB}{26,49,250}
\definecolor{SIj}{RGB}{52,178,51}
\definecolor{utblue}{RGB}{0, 0, 0}
\definecolor{utorange}{RGB}{90, 90, 90}

\usepackage{xspace}
\usepackage{amsfonts, amsmath}
\usepackage{mathtools}

\usepackage[ruled]{algorithm}
\usepackage[noend]{algpseudocode}
\algdef{SE}[SUBALG]{Indent}{EndIndent}{}{\algorithmicend\ }%
\algtext*{Indent}
\algtext*{EndIndent}

\newcommand{\eps}[0]{\ensuremath{\varepsilon}\xspace}
\newcommand{\sym}[0]{\ensuremath{\oplus}\xspace}

\newcommand{\dist}[0]{\ensuremath{c}\xspace}
\newcommand{\I}[0]{\ensuremath{\mathrm{I}}\xspace}
\newcommand{\SI}[0]{\ensuremath{\bar{\mathrm{I}}}\xspace}
\newcommand{\level}[1]{\ensuremath{\ell(#1)}\xspace}
\newcommand{\abs}[1]{\ensuremath{|#1|}\xspace}
\newcommand{\Mf}[0]{\ensuremath{M^{\mathrm{F}}}\xspace}
\newcommand{\Mon}[0]{\ensuremath{M'}\xspace}
\newcommand{\Mof}[0]{\ensuremath{M^*}\xspace}
\newcommand{\Mopt}[0]{\ensuremath{M^\OPT}\xspace}
\newcommand{\OPT}[0]{\ensuremath{\mathrm{OPT}}\xspace}
\newcommand{\IR}[0]{\ensuremath{\mathbb{R}}\xspace}

\newcommand{\talg}[0]{\mbox{$t$-net-cost algorithm}\xspace}
\def\interval{aggregate search interval}

\begin{document}
\pagestyle{plain}
\title{
Online Minimum Cost Matching on the Line with Recourse}
%
%
\author{Nicole~Megow \and 
Lukas~N\"{o}lke} 
\authorrunning{N. Megow and L. N\"{o}lke}
%
\institute{Department for Mathematics and Computer Science, University of Bremen, Germany \email{\{nmegow,noelke\}@uni-bremen.de}}%
\maketitle              
\begin{abstract}
In online minimum cost matching on the line, $n$ requests appear one by one and have to be matched immediately and irrevocably to a given set of servers, all on the real line. The goal is to minimize the sum of distances from the requests to their respective servers. Despite all research efforts, it remains an intriguing open question whether there exists an $O(1)$-competitive algorithm. The best known online algorithm by Raghvendra~\cite{Raghvendra18} achieves a competitive factor of~$\Theta(\log n)$. This result matches a lower bound of~$\Omega(\log n)$ \cite{DBLP:conf/latin/AntoniadisFT18} that holds for a quite large class of online algorithms, including all deterministic algorithms in the literature. 

In this work we approach the problem in a recourse model where we allow to revoke online decisions to some extent. We show an $O(1)$-competitive algorithm for online matching on the line that uses at most~$O(n\log n)$ reassignments. This is the first non-trivial result for min-cost bipartite matching with recourse. For so-called alternating instances, with no more than one request between two servers, we obtain a near-optimal result. We give a~$(1+\eps)$-competitive algorithm that reassigns any request at most~$O(\eps^{-1.001})$ times. 
This special case is interesting as the aforementioned quite general lower bound~$\Omega(\log n)$  holds for such instances.
\end{abstract}
\section{Introduction}\label{sec:intro} 

Matching problems are among the most fundamental problems in combinatorial optimization with great importance in theory and applications. In the bipartite matching problem, there is given a complete bipartite graph~$G=(R\cup S,E)$ 
with positive edge cost
$c_e$ for $e\in E$. The elements of~$R$ and~$S$ are called \emph{requests} and \emph{servers}, respectively, with~$n \coloneqq |R| \leq |S|$. A~matching~$M \subseteq S\times R$ is called \emph{complete} if every request in~$R$ is matched to a server in~$S$. The task is to compute a complete matching of minimum cost, where the \emph{cost} of a matching~$M$ is~$\dist(M) \coloneqq \sum_{e\in M} \dist_e$. If all information is given in advance, the optimum can be computed efficiently, e.g. by using the Hungarian Method~\cite{Kuhn55thehungarian}.

In an \emph{online setting}, the set of requests~$R$ is not known a-priori. Requests arrive online one by one and have to be matched immediately and irrevocably to a previously unmatched server. As we cannot hope to find an optimal matching, we use standard competitive analysis to evaluate the performance of algorithms. An online algorithm is~\emph{$\alpha$-competitive} if it computes for any instance a matching~$M$ with~$\dist(M) \leq \alpha \cdot \OPT$, where~$\OPT$ is the cost of an optimal matching. It is known that for arbitrary edge costs, the competitive ratio of any online algorithm is unbounded~\cite{KalyanasundaramP93,KhullerMV94}. For metric cost, there is a deterministic~$(2n-1)$-competitive algorithm known and this is optimal for deterministic online algorithms~\cite{KalyanasundaramP93,KhullerMV94}. A remarkable recent result by Nayyar and Raghvendra~\cite{NayyarR17} is a fine-grained analysis of an online algorithm based on {\em $t$-net cost}~\cite{Raghvendra16} showing a competitive ratio of~$O(\mu(G) \log^{2} n)$, where~$\mu(G)$ is the maximum ratio of the minimum TSP tour, and the weighted diameter of $G$. 

The online matching problem has so far resisted all attempts for achieving an $O(1)$-competitive algorithm even for special metric spaces such as the line. In {\em online matching on the line} the edge-costs are induced by a line metric; i.e., we identify each vertex of~$G$ with a point on the real line and the cost of an edge between a request and a server equals their distance on the line. The competitive ratio of the aforementioned algorithm is then $O(\log^{2} n)$, as~$\mu(S)=2$. This has been improved to $\Theta(\log n)$~\cite{Raghvendra18}, which is best possible for a large class of algorithms~\cite{AntoniadisBNPS19}. It remains a major open question whether there exists an $O(1)$-competitive online algorithm (deterministic or randomized) 
on the line. 

In this paper, we consider online matching on the line with \emph{recourse}. In the recourse model, we allow to change a small number of past decisions, i.e., at any point, we may delete a set of edges~$(r_i,s_i)$ of the current matching and match the requests~$r_i$ to different (free) servers. Online optimization with recourse is an alternative model to standard online optimization which has received increasing popularity recently. Obviously, if the recourse is not limited then one can just simulate an optimal offline algorithm, and the online nature of the problem disappears. The challenging question for online matching on the line is whether one can maintain an~$O(1)$-competitive solution with bounded recourse. 

\medskip 
\noindent \textbf{Our results.} 	We give the first\footnote{We remark that after submission of our work, we learnt about a simultaneously submitted preprint with results on our problem of online min-cost matching with recourse on the line and in general metrics \cite{GuptaKS19}. Their work obtains similar and more general results using completely different techniques.} non-trivial results for the min-cost online matching problem with recourse. We show that with limited recourse, one can indeed maintain an~$O(1)$-competitive solution. 
\begin{theo}\label{thm:general}
	The online bipartite matching problem on the line admits a constant competitive algorithm with amortized recourse budget~$O(n \log n)$.
\end{theo}
Our algorithm builds on the~$t$-net-cost algorithm by Raghvendra~\cite{Raghvendra16,Raghvendra18}. It
has the nice property that it interpolates by parameter~$t$ between an~$O(\log n)$-competitive online solution (without recourse) and an~$O(1)$-approximate offline solution (with large recourse). To obtain an~$O(1)$-competitive algorithm with~$O(n \log n)$ recourse, we maintain a dynamic partition of the edge set and determine an online matching on one type of (active) edges and an \mbox{approximate} offline matching on the other type of (inactive) edges. To further limit the amount of recourse actions on inactive edges, we incorporate an edge freezing scheme. 

We investigate also the special class of instances, called \emph{alternating} instances, where between any two requests on the line there exists at least one server. This is an interesting class 
as the quite strong lower bound of~$\Omega(\log n)$ for all deterministic online algorithms {\em without recourse} studied in literature~\cite{DBLP:conf/latin/AntoniadisFT18}, holds even on alternating instances. For alternating instance, we present a more direct and near-optimal algorithm with a scalable performance-recourse tradeoff. 
\begin{theo}\label{thm:alternating}
 For alternating instances of online matching on the line, there is a~$(1+\eps)$-competitive algorithm that reassigns a request $O(\eps^{-1.001})$ times.
\end{theo}
While the algorithm is quite simple, the proof requires a novel charging scheme that exploits the special structure of optimal solutions on alternating instances. We observe that a large number of recourse actions for a specific request points to large edges in the optimal solution elsewhere on the line. 

As a byproduct we give a simple analysis of (a variant of) the algorithm in the traditional online setting without recourse. We show that it is $O(\log \Delta)$-competitive for alternating requests on the line, where $\Delta$ is some parameter depending on the cost-metric. 
This result compares to the competitive ratio $O(\log n)$ in~\cite{Raghvendra18} for the general line.

\medskip 
\noindent \textbf{Further related work.} Extensive literature is devoted to online bipartite matching problems. The maximum matching variant is quite well understood. For the unweighted setting optimal deterministic and randomized algorithms with competitive ratio $2$ and $e/(e-1)$ are known~\cite{KarpVV90}. The weighted maximization setting does not admit bounded guarantees in general, but intensive research investigates  models with additional assumptions; see, e.g., the survey~\cite{Mehta13}. The online min-cost matching problem is much less understood. 
It remains a wide open question whether a constant-competitive algorithm, deterministic or randomized, is possible for online min-cost matching on the line. In fact, the strongest known lower bound is $9+\eps$~\cite{DBLP:journals/tcs/FuchsHK05}. For a quite large class of algorithms, including all  deterministic ones in the literature, there is lower bound of~$\Omega(\log n)$~\cite{DBLP:conf/latin/AntoniadisFT18}.
 
Randomization in the algorithms decisions improves upon the best possible deterministic competitive ratio of~$(2n-1)$ for metric online bipartite matching~\cite{KalyanasundaramP93,KhullerMV94}; there is an~$O(\log^{2} n)$-competitive randomized algorithm~\cite{BansalBSN2014}. On the line, no such improvement on the deterministic result by randomization is known; the competitive factor of $O(\log n)$ is the best known result for both, deterministic and randomized, algorithms~\cite{Raghvendra18,GuptaL12}. 
	
Interestingly, when assuming randomization in the order of request arrivals (instead of an adversarial arrival order), the natural Greedy algorithm is~$n$-competitive~\cite{GairingK19} for general metric spaces. Interestingly, the online {\em $t$-net cost} algorithm is in that case $O(\log n)$-competitive~\cite{Raghvendra16}. Very recently, Gupta et al.~\cite{GuptaGPW19} gave an $O((\log\log \log n)^2)$-competitive algorithm in the model with online known i.i.d. arrivals.

Maintaining an online cardinality-maximal bipartite matching with recourse was studied extensively~\cite{BernsteinHR19,BosekLSZ14,ChaudhuriDKL09,GroveKKV95,AngelopoulosDJ18}. Bernstein et al.~\cite{BernsteinHR19} showed recently that the $2$-competitive Greedy requires amortized $O(n\log n^2)$ reassignments, which leaves a small gap to the lower bound of $\Omega(n \log n)$. 
We are not aware of any results for the min-cost variant, not even for matching on the line. It remained open whether recourse can improve upon the performance bound of~$O(\log n)$~\cite{Raghvendra18}.

The following two models address other types of matching with recourse. In a setting motivated by scheduling, several requests can be matched to the same server and the goal is to minimize the maximum number of requests assigned to a server. Gupta et al.~\cite{GuptaKS14} achieve an $O(1)$-competitive ratio with amortized $O(n)$ edge reassignments.
A quite different two-stage robust model has been proposed recently by Matuschke et al.~\cite{MatuschkeSV19}. In a first stage one must compute a perfect matching on a given graph and in a second stage a batch of~$2k$ new nodes appears which must be incorporated into the first-stage solution to maintain a low-cost matching by reassigning only few edges. For matching on the line, they give an algorithm that maintains a $10$-approximate matching reassigning $2k$ edges.

Recourse in online optimization has been investigated also for other min-cost problems even though less than for maximization problems. Most notably seems the online Minimum Steiner  tree problem \cite{GuG016,ImaseW91,MegowSVW16,LackiOPS2015}. Here, one edge reassignment per iteration suffices to maintain an $O(1)$-competitive algorithm~\cite{GuG016}, whereas the online setting without recourse admits a competitive ratio of $\Omega(\log n)$.

The recourse model has some relation to {\em dynamic algorithms}. Instead of minimizing the number of past decisions that are changed (recourse), the dynamic model focusses on the running time to implement this change ({\em update time}). 
A full body of research exists on maximum (weighted) bipartite matching; we refer to the nice survey in~\cite{DuanP14}. We are not aware of any results for min-cost matching.

\section{Preliminaries and overview of techniques}\label{sec:overview}

A common method for increasing the cardinality of an existing matching~$M$ is to augment along an alternating path~$P$. After augmentation, the resulting matching $\tilde M$ is given by the symmetric difference\footnotemark \mbox{$M \sym P$}.%
\footnotetext{For two sets~$X,Y\!$, their symmetric difference is given by~$X \sym Y \coloneqq (X\cup Y) \setminus (X\cap Y)$. For matchings~$M_1, M_2$, their symmetric difference~$ \;M_1 \sym M_2$ consists of disjoint alternating paths and cycles.} 
There might be a choice between different augmenting paths; typically, a path of minimum cost~(w.r.t.\ some metric) is selected. 
Recently, Raghvendra~\cite{Raghvendra16} introduced the following metric. 
For~$t>1$, the \emph{$t$-net-cost} of a path~$P$ with respect to a matching~$M$ is 
\begin{align}
\phi_t^{M}(P) \coloneqq \; t \cdot c(P\setminus M)  - c(P\cap M) \ = \ t \cdot c(P\cap \tilde M)  - c(P\cap M). \label{eq:tnetcost}
\end{align}

Our algorithm maintains three matchings at the same time: the \emph{recourse matching}~$M$, which is the actual output of the algorithm, and two auxiliary matchings based on (online and offline versions of) the \talg\cite{Raghvendra16}, namely, the \emph{offline matching}~$\Mof$ and the \emph{online matching}~$\Mon$. While~\Mof is a near-optimal offline matching that possibly requires a large number of recourse actions,~\Mon is an online matching 
that is only~$O(\log n)$-competitive~\cite{Raghvendra18} but uses no recourse.

Matchings \Mof and~\Mon are obtained based on the above cost function; see also\cite{Raghvendra16,NayyarR17,Raghvendra18}. Let a matching indexed by $i$ denote its state after serving the $i$-th request including possible reassignments. We refer to this as the state at \emph{time $i$}.

The offline \talg constructs~$\Mof_i$ by augmenting~$\Mof_{i-1}$ along an alternating path~$P_i$ of minimum $t$-net-cost w.r.t.\ $\Mof_{i-1}$, that is,~\mbox{$\Mof_{i} \coloneqq~\Mof_{i-1} \sym P_i$}. The path $P_i$ starts at~$r_i$ and ends at a free server,
denoted by~$s_i$. While this procedure may rematch requests very often, the cost of the resulting matching for any~$i$ and~$t>1$ satisfies~$\dist(\Mof_i) \leq t\cdot\OPT_i$ as shown in~\cite{Raghvendra16}.
Here, $\OPT_i \coloneqq c(\Mopt_i)$ denotes the cost of an optimal offline matching~$\Mopt_i$ for the first~$i$ requests.

The online \talg maintains $\Mof$ as auxiliary matching. It constructs the online matching~$\Mon_i$ by directly connecting the end points~$r_i$ and~$s_i$ of~$P_i$, where $P_i$ is the augmenting path of minimum $t$-net-cost w.r.t.~$\Mof_{i-1}$. That is, we set $\Mon_i \coloneqq~\Mon_{i-1} \cup \{(r_i,s_i)\}$ without any recourse.

It is interesting to note that the parameter $t$  allows to interpolate between different known algorithms. When~$t=1$, the offline \talg  is equivalent to the Hungarian Method~\cite{Kuhn55thehungarian} which computes an optimal solution. The corresponding online matching, however, can easily be seen to have a competitive ratio of $\Omega(n)$. In contrast, when~$t$ tends to infinity, the algorithms behavior resembles the greedy online algorithm which matches a request on arrival to the nearest free server. Its competitive ratio can even be exponential~\cite{KalyanasundaramP93}. Interestingly though, with~$t=3$, is has a competitive ratio of~$O(\log n)$~\cite{Raghvendra18}.

We exploit this fact in the design of an algorithm for the recourse model by constructing a matching that is carefully balanced between~\Mof and~\Mon. This allows us to simultaneously bound competitive ratio and number of reassignments.
When a request arrives, we match it as in~\Mon and locally 
group it with other recent requests.
This step does not require any recourse action 
and can be interpreted as an input-buffer for~\Mof. 
Once the total cost of a group of requests 
grows too large, we release the requests and 
re-match them as in~\Mof.
As this may require too many edge reassignments, we freeze certain low-cost edges and limit the resulting increase in the cost of other edges.

For alternating instances, we again maintain~\Mof and employ a simple freezing scheme.
After a request reaches a certain threshold of reassignments, we freeze this request and the currently associated matching edge.
We charge detours that are taken due to frozen edges
to large non-frozen edges of~\Mof.

\section{A Constant-Competitive Algorithm}\label{sec:generalALG}

Our algorithm classifies requests according to the structure of intervals representing the part of the line where, on arrival of a request, the \talg searches for a free server. 
Define the \emph{search interval} of a request~$r_i$ as the open interval~$\SI_i = (s_i^\mathrm{L},s_i^\mathrm{R})$, where~$s_i^\mathrm{L}$ and~$s_i^\mathrm{R}$ are virtual servers on the line (not necessarily in~$S$) farthest to the left and right of~$r_i$, respectively, that can be reached from~$r_i$ with~$t$-net-cost~$\phi_t^{\smash{\Mof_{i-1}}}(P_i)$. 
In other words,~$\SI_i$ is the convex hull of all points on the line, reachable from~$r_i$ with~$t$-net-cost (strictly) less than~$\phi_t^{\smash{\Mof_{i-1}}}(P_i)$.
We define the \emph{\interval} of~$r_i$ to be the maximal interval~$\I_i$ which contains~$r_i$ and is a subset of~$\bigcup_{j\leq i}\SI_j$; see Figure~\ref{fig:blocks} for an illustration.
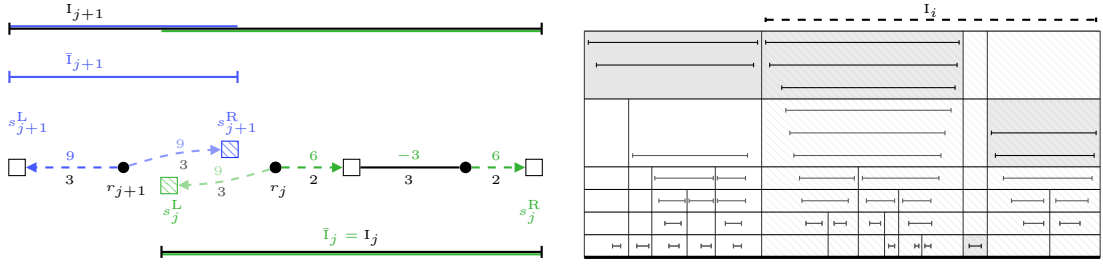
\begin{figure}[t]
	\centering
	\begin{tikzpicture}[xscale=.5, yscale = .8,>={Stealth[inset=0pt,length=4pt,angle'=50]}]	
		\foreach \x / \name in {1.2/A,10/B,14.8/C}{	
			\node[server, draw = black](\name) at (\x,0) {};
		}			
		\node[server, draw = utgreen!0, fill = SIj!50, schraffiert, pattern color=SIj!50, thin](V) at (5.2,-.3) {};	
		\node[server, SIj, thin, fill opacity = 0]() at (5.2,-.3) {};	
		\node[server, draw = utgreen!0, fill = SIjp1!50, schraffiert, pattern color=SIjp1!50, thin](V2) at (6.8,.3) {};	
		\node[server, SIjp1, thin, fill opacity = 0] at (6.8,.3) {};	
		\node[anchor = north ] () at (14.7,-0.4) {\color{SIj}\tiny$s_j^\mathrm{R}$};
		\node[anchor = south ] () at (1.5,0.4) {\color{SIjp1!80}\tiny$s_{j+1}^\mathrm{L}$};
		\node[anchor = south ] () at (7,.4) {\color{SIjp1!80}\tiny$s_{j+1}^\mathrm{R}$};
		\node[anchor = north ] () at (5.3,-0.4) {\color{SIj}\tiny$s_j^\mathrm{L}$};
		\foreach \x / \name in {8/a,13/b}{	
			\node[request, draw opacity = 0](\name) at (\x,0) {};
		}			
		\node[] () at (4,-.4) {\tiny~$r_{j+1}$};
		\node[] () at (2.5,-.2) {\tiny~$3$};
		\node[] () at (5.5,0.05) {\color{black!70}\tiny~$3$};	
		\node[] () at (6.5,-.4) {\color{black!70}\tiny~$3$};	
		\node[request,scale = .98](r) at (4,0) {};
		\node[request,scale = .98](c) at (8,0) {};	
		\draw[matching](b) -- (B);
		\draw[augment, dashed, SIj, thick, ->](a) -- (B);
		\draw[augment, dashed, SIj, thick, ->](b) -- (C);
		\draw[augment, SIj!50, dashed, thick, ->](a) to [bend left = 5] (V);
		\draw[augment, SIjp1!50, dashed, thick, ->](r) to [bend left = 5] (V2);
		\draw[augment, dashed, SIjp1!80, thick, ->](r) -- (A);
		\node[] () at (2.6,.2) {\color{SIjp1!80}\tiny$9$};
		\node[] () at (5.5,.4) {\color{SIjp1!50}\tiny$9$};	
		\node[] () at (6.5,-0.05) {\color{SIj!50}\tiny$9$};	
		\draw[SIjp1!80,thick](1,1.5) -- (7,1.5);		
		\draw[SIjp1!80,thick] (1,1.4) -- (1,1.6);
		\draw[SIjp1!80,thick] (7,1.4) -- (7,1.6);
		\node[] () at (3,1.75) {\color{SIjp1!80}\tiny$\SI_{j+1}$};
		\draw[SIjp1!80,thick](1,2.34) -- (7,2.34);	
		\draw[SIj, thick](5,2.26) -- (15,2.26);
		\draw[thick](1,2.3) -- (15,2.3);	
		\draw[thick] (1,2.2) -- (1,2.4);
		\draw[thick] (15,2.2) -- (15,2.4);
		\node[] () at (3,2.57) {\tiny$\I_{j+1}$};	
		\foreach \x / \labela / \labelb / \col in {9/6/2/SIj,11.5/-3/3/SIj,13.8/6/2/SIj}{	
			\node[] () at (\x,.2) {\color{\col}\tiny$\labela$};
			\node[] () at (\x,-.2) {\tiny$\labelb$};
		}	
		\node[] () at (8,-.4) {\tiny~$r_j$};
		\node[] () at (10,-1.15) {\color{SIj}\tiny$\SI_j=\color{black}\I_j$};
		\draw[SIj, thick](5,-1.44) -- (15,-1.44);
		\draw[thick](5,-1.4) -- (15,-1.4);
		\draw[thick] (5,-1.3) -- (5,-1.5);
		\draw[thick] (15,-1.3) -- (15,-1.5);
	\end{tikzpicture}
	\hfill 
	\begin{tikzpicture}[yscale = .6, xscale=.53]	
	\fill[utblue!10, thick] (3.3,3.75) rectangle (12.7,5.25);
	\fill[utblue!10, thick] (13.3,2.25) rectangle (16.1,3.75);
	\fill[utblue!10, thick] (12.7,0.25) rectangle (13.3,.75);	
	\fill[schraffiert, pattern color=black!5 , thick] (7.7,0.25) rectangle (16.1,5.25);
	\foreach \x / \lev  in {3.3/11, 4.4/8, 4.975/5, 5.85/4, 6.55/5, 7.7/11, 9.35/3, 10.1/5, 10.75/3, 11.1/4, 11.675/2 ,12.7/11, 13.3/11, 14.85/4, 16.1/11}{	
		\draw[utblue!80, very thin] (\x,0.25) -- (\x,\lev*.5-.25);
	}
	\draw[very thick] (3.297,0.25) -- (16.103,0.25);
	\foreach \a / \b / \lev / \col / \cov in {4/4.2/1/utgreen/7-, 4.6/4.9/1/utgreen/7-, 5.4/5.65/1/utgreen/1-, 6.2/6.45/1/utgreen/3-, 7/7.2/1/utgreen/7-, 10.85/11/1/utgreen/7-, 11.5/11.6/1/utgreen/7-, 11.75/11.9/1/utgreen/7-, 12.85/13.15/1/utblue/7-, 
		5.3/5.7/2/utgreen/2-, 6.8/7.2/2/utgreen/7-, 8.8/9.2/2/utgreen/7-, 9.5/9.8/2/utgreen/7-, 10.38/10.65/2/utgreen/7-, 11.35/11.92/2/utgreen/7-, 14.2/14.7/2/utgreen/7-, 15.1/15.6/2/utgreen/7-, 
		5.15/5.8/3/utorange/4-, 5.9/6.5/3/utorange/5-, 6.6/7.3/3/utorange/7-, 8.7/9.85/3/utorange/7-, 10.3/11/3/utorange/7-, 11.2/11.9/3/utorange/7-, 13.9/14.7/3/utorange/7-, 15/15.8/3/utorange/7-, 
		5.05/6.5/4/utorange/6-, 6.6/7.3/4/utorange/7-, 8.64/10/4/utorange/7-, 10.2/12.1/4/utorange/7-, 13.7/15.9/4/utorange/7-, 
		4.5/7.35/5/utorange/7-, 8.5/12.15/5/utorange/7-, 13.5/15.95/5/utblue/7-, 
		8.4/12.25/6/utorange/7-, 13.4/16/6/utblue/7-, 
		8.3/12.4/7/utorange/7-, 
		8.2/12.5/8/utblue/7-, 
		3.6/7.5/9/utblue/7-, 7.9/12.55/9/utblue/7-, 
		3.4/7.6/10/utblue/7-, 7.8/12.6/10/utblue/7-, 
		4/4.2/1/utgreen/7-, 4.6/4.9/1/utgreen/7-, 5.4/5.65/1/utgreen/1-, 6.2/6.45/1/utgreen/3-, 7/7.2/1/utgreen/7-, 10.85/11/1/utgreen/7-, 11.5/11.6/1/utgreen/7-, 11.75/11.9/1/utgreen/7-, 
		5.3/5.7/2/utgreen/2-, 6.8/7.2/2/utgreen/7-, 8.8/9.2/2/utgreen/7-, 9.5/9.8/2/utgreen/7-, 10.38/10.65/2/utgreen/7-, 11.35/11.92/2/utgreen/7-, 14.2/14.7/2/utgreen/7-, 15.1/15.6/2/utgreen/7- 
	}{	
		
		\draw[\col] (\a,\lev*.5) -- (\b,\lev*.5);
		\draw[thin, \col] (\a,\lev*.5 + .05) -- (\a,\lev*.5 -.05);
		\draw[thin, \col] (\b,\lev*.5 + .05) -- (\b,\lev*.5 -.05);
	}		
	\foreach \lev / \col  in {1/, 2/utgreen, 3/, 4/utnavy, 7/utnavy, 10/utnavy }{	
		\draw[thin, black] (3.3,\lev*.5 + .25) -- (16.1,\lev*.5 + .25);
	}
	\draw[thick, dashed, black] (7.8,11*.5) -- node[above = -2.5 pt] {\tiny$\I_i$} (16,11*.5);
	\draw[thick, black] (7.8,11*.5 - .05) -- (7.8,11*.5 + .05);
	\draw[thick, black] (16,11*.5 - .05) -- (16,11*.5 + .05);
	\end{tikzpicture}\hspace{.15 cm}
	\caption{Left: Construction of search 
	and \interval s of requests~$r_j$ and~$r_{j+1}$.
		Right: Typical block structure, $\I_j$ opens a new top block (active edges) and initiates a recourse step. Requests in the hatched area (inactive) are reassigned.
	}\label{fig:blocks}
	\vspace*{-2ex}
\end{figure}

\begin{obs} 
	Whenever~$i<j$, then either~$\I_i\cap \I_j = \emptyset$ or~$\I_i \subseteq \I_j$. \label{obs:laminar}
\end{obs}
We say an \interval~$\I_i$ has \emph{level}~$k$, if~$(1+\eps)^{k-1} \leq \abs{\I_i} < (1+\eps)^k$ and write~$\level{\I_i}=k$. Throughout the algorithm, we set~$t=3$ and~$\eps = \tfrac{1}{32t}$.
Further, we say that two \interval s of the same level which intersect with each other belong to the same \emph{block}.
If the \interval s of a block do not intersect those of higher level, then this block is said to be a \emph{top block}.
We~may say that a request~$r_i$ or the edge~$e_i = (r_i,s_i) \in M'$ is of level~$k$ or belongs to a certain block if this holds for the associated \interval.
A~typical block-structure is depicted in Figure~\ref{fig:blocks}.

Our algorithm crucially uses the structure imposed by the blocks and partitions the requests into three groups. 
The first group consists of all requests that belong to a top block. We label requests in this group \emph{active}. All other requests are \emph{inactive}. An active request~$r_i$ is matched exactly as in the online matching, i.e., it is matched in~$M_i$ via the edge~$e_i = (r_i,s_i) \in \Mon_i$ which connects the endpoints of the
minimum $t$-net-cost augmenting path w.r.t.\ $\Mof_{i-1}$. 

The second group of requests is also matched according to~\Mon.
It consists of requests~$r_i$, for which the corresponding edge~$e_i \in~\Mon$ is of negligible size.
We call such a request or edge \emph{frozen} and denote by~$\Mf \subseteq \Mon$ the set of frozen edges. The precise freezing scheme (Algorithm~\ref{alg:general}, Steps 1-3) is detailed below. 
It will ensure that frozen requests contribute at most \OPT to the total cost of~$M$ and, intuitively, it is not worth to spend recourse 
on them.

Ideally, we would like to match the third group, i.e.\ the remaining \emph{non-frozen, inactive} requests, as in the offline matching~\Mof. This may not be possible, as \Mof could assign a request~$r$ to a server $s$ that is already used by a frozen request. In such a case, we match~$r$ via a detour of low additional cost as follows; we call this {\em detour matching}. 
We consider the symmetric difference~$\Mof \sym \Mf$ which consists of alternating paths and cycles.
Since~$r$ is not frozen, there exists a (unique) alternating path, that starts with the edge~$(r,s)$ and ends at some server~$s'$ that is not yet matched to a frozen request.
In the recourse matching, we match~$r$ directly to this server~$s'$~(Alg.~\ref{alg:general}, Step~6).
For an illustration, see Fig.~\ref{fig:unfreezing}, left part. 

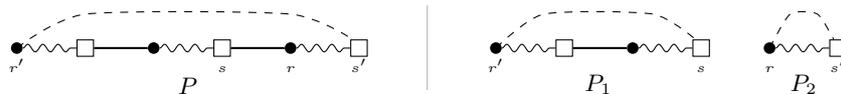
\begin{figure}[tbh]
	\vspace*{-2ex}
	\centering
	\begin{tikzpicture}[scale = .45]
	\node[]() at (-7,0) {	\begin{tikzpicture}[scale = .9, >={Stealth[inset=0pt,length=4pt,angle'=50]}]
		\draw[thin, dashed] plot [smooth, tension = .4] coordinates { (1,0) (2,.5) (5,.5) (5.95,0.1)};	
		\foreach \x / \name in {2/A,4/B,6/C}{	
			\node[server, fill= white, draw = black](\name) at (\x,0) {};
		}		
		\node[anchor = south ] () at (4,-0.48) {\tiny$s$};
		\node[anchor = south ] () at (6,-0.48) {\tiny$s'$};
		\node[anchor = south ] () at (3.5,-0.8) {\small$P$};
		\foreach \x / \name in {1/a,3/b,5/c}{	
			\node[request, draw opacity = 0](\name) at (\x,0) {};
		}			
		\node[anchor = south ] () at (1,-0.48) {\tiny$r'$};
		\node[anchor = south ] () at (5,-0.48) {\tiny$r$};
		\foreach \A / \B  in {a/A,b/B,c/C}{
			\draw[snake it, thin](\A) -- (\B);
		}	
		\foreach \A / \B  in {A/b,B/c}{
			\draw[thick](\A) -- (\B);
		}	
		\end{tikzpicture}};
	\draw[black!30, thin](0,-1) -- (0,1.5);
	\node[]() at (7,0) {	\begin{tikzpicture}[scale = .9, >={Stealth[inset=0pt,length=4pt,angle'=50]}]
		\draw[thin, dashed] plot [smooth, tension = .4] coordinates { (1,0) (1.7,.5) (3.3,.5) (3.95,0.1)};
		\draw[thin, dashed] plot [smooth, tension = .4] coordinates { (5,0) (5.3,.5) (5.7,.5) (5.95,0.1)};		
		\foreach \x / \name in {2/A,4/B,6/C}{	
			\node[server, fill= white, draw = black](\name) at (\x,0) {};
		}		
		\node[anchor = south ] () at (4,-0.48) {\tiny$s$};
		\node[anchor = south ] () at (6,-0.48) {\tiny$s'$};
		\node[anchor = south ] () at (2.5,-0.8) {\small$P_1$};
		\node[anchor = south ] () at (5.5,-0.8) {\small$P_2$};
		\foreach \x / \name in {1/a,3/b,5/c}{	
			\node[request, draw opacity = 0](\name) at (\x,0) {};
		}			
		\node[anchor = south ] () at (1,-0.48) {\tiny$r'$};
		\node[anchor = south ] () at (5,-0.48) {\tiny$r$};
		\foreach \A / \B  in {a/A,b/B,c/C}{
			\draw[snake it, thin](\A) -- (\B);
		}	
		\foreach \A / \B  in {A/b}{
			\draw[thick](\A) -- (\B);
		}	
		\end{tikzpicture}};
\end{tikzpicture}
\caption{Illustration of the third step in Algorithm~1. Edges of~\Mof are drawn as snaked lines, edges of~\Mf as solid lines and edges of~$M\setminus\Mf$ dashed.
	After unfreezing~$r$, the path in~$\Mof \sym \Mf$ is split in two parts by the removal of~$(r,s)$.}\label{fig:unfreezing}
	\vspace*{-3ex}
\end{figure}

It remains to specify how to obtain the desired matching $M$ described above. When a request arrives, it belongs immediately to a top block (active), by definition of \interval  s, and is thus assigned according to $M'$~(Algorithm~\ref{alg:general}, Step~7). In the course of the algorithm's execution new blocks may appear (on top) rendering previously active requests inactive.
This requires a batch of recourse actions to maintain~$M$.
Further, we describe the precise freezing scheme. Note that we cannot verify online whether an edge has cost $c_e\leq \text{OPT}/n$ since we do not know OPT or $n$. We need a dynamic freezing and unfreezing scheme depending on the known value OPT$_i$.
\begin{algorithm}[thb]
\caption{}\label{alg:general} 
\begin{algorithmic}[1]
	\Statex \textbf{On arrival} of the~$i$-th request~$r_i$: \vspace{6 pt}
	\Statex \textbf{Freezing/Unfreezing:}
	\Indent
	\State update the set~\Mf of frozen edges
	\For{edges that become unfrozen}
	\State repair assignments on corresponding path of~$\Mof \sym \Mf$
	\EndFor\vspace{4 pt}
	\EndIndent
	\Statex \textbf{Recourse step} (new top block):
	\Indent
	\If{there is no~$j < i$ such that~$\I_j \subseteq \I_i$ and~$\level{\I_j} = \level{\I_i}$}
	\For{non-frozen requests~$r \in \I_i$}
	\State reassign~$r$ according to~$\Mof_{i-1} \sym \Mf_{i-1}$
	\EndFor 
	\EndIf
	\EndIndent
	\State match~$r_i$ to~$s_i$, as in~$\Mon_i$
\end{algorithmic}
\end{algorithm}

\noindent \emph{Recourse step:} 
A new top block is created whenever a request~$r_i$ arrives and there is no~$j < i$ such that~$\I_j \subseteq \I_i$ and~$\level{\I_j} = \level{\I_i}$ hold. 
This triggers, what we call a \emph{recourse step} (Alg.~\ref{alg:general}, Steps~5-6) involving a number of recourse actions to accommodate the newly inactive requests. 
To that end, they and all other non-frozen, inactive requests that lie in~$\I_i$ are reassigned. 
They are matched according to~$\Mof_{i-1} \sym \Mf_{i-1}$ which we described as detour matching above; see Figure~\ref{fig:blocks}. 

\smallskip 
\noindent \emph{Freezing/Unfreezing}: 
When a new request $r_i$ arrives, we update the set of frozen edges as follows. If the cost of an edge~$e \in M'$ drops below~$\tfrac{\OPT_i}{i^2}$, we freeze this edge (and the corresponding request). 
The edge stays frozen until, at a later time~$j$, we have~\smash{$\dist_e > \tfrac{\OPT_j}{j}$} at which point we unfreeze it.
As the edge becomes now too costly to be matched as in~\Mon, we need to reassign it according to~$\Mof \sym \Mf$ as follows. 
Consider an edge~$e=(s,r)$ right before it is unfrozen. If~$e\in \Mof$, then there is nothing to do. Otherwise,
$e$ is an interior edge of some alternating path~$P$ in~$\Mof \sym \Mf$ that starts in a request~$r'$ and ends in a server~$s'$.
Unfreezing~$e$ and matching it according to~$\Mof \sym \Mf$ decomposes~$P$ into the~\mbox{$r'$-$s$-path~$P_1$} and the~$r$-$s'$-path~$P_2$.
In Algorithm~1 we implement these changes via two recourse actions: we reassign~$r$ to~$s'$ and~$r'$ to~$s$; see Figure~\ref{fig:unfreezing}.

\medskip 
\noindent 
\textbf{Bounding the competitive ratio.} We recall that $\Mof \leq t\cdot\OPT$~\cite{Raghvendra16}. Frozen edges satisfy \mbox{$c(\Mf) \leq \OPT$} as, after arrival of the last request, there are at most~$n$ frozen edges of cost at most $\frac{\OPT}{n}$ each. Therefore, the contribution of inactive non-frozen requests to the cost of~$M$ is at most~\mbox{$\dist(\Mof \sym \Mf) \leq (t+1)\cdot\OPT$}.
\begin{obs}\label{obs:inactiveEdges}
	The cost of inactive edges is at most~$(t+2)\cdot\OPT$.
\end{obs}
To bound the cost of active edges, we build on the analysis of Raghvendra~\cite{Raghvendra18}.
We refine his technical propositions 
and perform 
a slightly more fine-grained analysis. 
	Instead of simultaneously bounding the cost of all blocks of the same level, we argue more generally on the cost of any set of disjoint blocks at different levels.
	In particular, we are interested in bounding the cost contribution of the top blocks.
\begin{lemm}\label{lem:activeEdges}
	For~$t=3$, the cost of all active edges of~$M$ is bounded by~$O(\OPT)$.
\end{lemm}
\begin{cor}\label{cor:comp}
	Algorithm~1 has a constant competitive ratio.
\end{cor}
We outline the proof of Lemma~\ref{lem:activeEdges}, omitting technical details. To this end, we use a refined version of the analysis in~\cite{Raghvendra18} adapted to our notion of blocks.
For a full account, see~\cite{Raghvendra18} with this outline as reference.
While our definition of search intervals is intuitive and useful (specifically for the proof of Lemma~\ref{lem:freeze} later on), it describes intervals different from the search intervals in~\cite{Raghvendra18}.
However, we show that our definition of \interval s coincides with Raghvendra's definition of \emph{intervals of a cumulative search region}, which we denote here by~$C$. 

Both intervals~$C$ and~$\I$ are constructed the same way, that is, by taking the union of all known search intervals (for the respective definition) and choosing of the resulting new intervals the one which contains the considered request.
We describe the definition of search intervals in~\cite{Raghvendra18}, which, to avoid confusion, we call \emph{dual intervals}.
The~$t$-net-cost algorithm maintains dual values \mbox{$y: S \cup R \rightarrow \IR^+$} satisfying~$y_s + y_r = \dist_{(s,r)}$ if~$(s,r)\in\Mof$ and~$y_s + y_r \leq t\cdot \dist_{(s,r)}$ otherwise.
Additionally, duals of free requests or servers are zero.
When a request~$r_i$ arrives, a shortest~$t$-net-cost path~$P_i$ is found and the duals of all vertices in the search tree (denoted by~$A_i \subseteq S$ and~$B_i \subseteq R$) are updated before augmentation so that the dual constraints on~$P_i$ (and also on~$P_i^\mathrm{L}$/$P_i^\mathrm{R}$) are tight.
A dual interval of a request~$r_i$ is defined as~${interior}(\bigcup_{r\in B_i} cspan(r,i))$, 
where~$cspan(r,i)= \smash{[r - \tfrac{y_{\max}^i(r)}{t},r + \tfrac{y_{\max}^i(r)}{t}]}$ and~$y_{\max}^i(r)$ is the highest dual weight assigned to~$r$ until time~$i$.

Raghvendra~\cite{Raghvendra16} shows, that the dual constraints on~$P_i^\mathrm{L}$ and~$P_i^\mathrm{R}$ are tight before augmentation.
Therefore, we obtain our search intervals~$\I_i$ by replacing~$y_{\max}^i(r)$ with the dual weight of~$r$ before augmentation along~$P_i$. 
Therefore, clearly~$\I_j \subseteq C_j$.
At the same time, dual weights of requests are increased when the request is contained in some~$B_i$ as detailed above or reduced right after an augmentation.
Thus, the maximal value~$y_{\max}^i(r)$ is attained right before an augmentation, in which case the request is part of some~$P_h^\mathrm{L}$ or~$P_h^\mathrm{R}$, implying~$C_j \subseteq \I_j$.
Hence the intervals $\I_i$ and $C_i$ coincide and we may use the respective results from~\cite{Raghvendra18}.

An edge~$e_i$ of the online matching~$\Mon$ is called \emph{short}, if the corresponding augmenting path~$P_i$ satisfies~$\ell(P_i) \le \tfrac{4}{t-1}\phi_t^{{i-1}}(P_i)$ and \emph{long} otherwise.
Raghvendra bounds the 
cost of long edges by that of short ones.
We give a more general statement. 
\begin{lemm}\label{lem:ConsecutivePaths} 
	Let~$P_j,P_{j+1},\ldots,P_k$ be augmenting 
	paths that are consecutively used in Algorithm~1 and~$e_j,e_{j+1},\ldots,e_k$ the corresponding edges in~\Mon. Then
	\begin{align}
	\smash[t]{\sum_{\smash{i:\,e_i\textnormal{ long}}}} \dist(e_i) \leq \smash[t]{ \big(3+\tfrac{4}{t-1}\big)\cdot\sum_{\smash{i:\,e_i\textnormal{ short}}} \dist(e_i) + 4\cdot \OPT_{j-1}.}
	\end{align}
\end{lemm}
\begin{proof}[sketch]
	Raghvendra~\cite[Proof of Lemma 2]{Raghvendra18} shows (with~$j=1$ and~$k=n$) that
	$
	\tfrac{t-1}{2}\sum_{i=j}^k \ell(P_i) =
	\sum_{i=j}^k\phi_t(P_i) - \tfrac{t+1}{2} \sum_{i=j}^k(\dist(\Mof_i) - \dist(\Mof_{i-1})). 
	$
	We obtain the statement of the lemma by slightly altering the subsequent steps of the proof.
	Using~$\OPT_{j-1} \leq \OPT_k \leq \dist(\Mof_k)$ and~$\dist(\Mof_{j-1}) \leq t\cdot  \OPT_{j-1}$, we get 
	\begin{align*}
	\textstyle\tfrac{t-1}{2}\sum_{i=j}^k \ell(P_i) \leq (t-1) \cdot  \sum_{i=j}^k\phi_t(P_i) + (t-1)\cdot\OPT_{j-1} .
	\end{align*}
	Continuing along the lines of the proof \cite{Raghvendra18}, one easily obtains the statement.
	\qed  
\end{proof}

In~\cite[Lemma~13]{Raghvendra18} level-$k$ blocks are considered as separate instances.
Denote by~$\mathcal{B}_i$ the instance corresponding to block~$B_i$.
It consists of the requests~$r_j \in B_i$ as well as the corresponding servers~$s_j$. 
The sum of costs of the individual optimal solutions can be bounded 
by~$2t\cdot \OPT$. 
Since the proof 
only uses the disjointness of level-$k$ blocks, we may use the verbatim proof to obtain the following.
\begin{lemm}\label{lem:boundBlocks}
	For 
	disjoint blocks $B_1,B_2,\ldots,B_h$, we have
	$\sum_{i=1}^{h} \OPT_{\mathcal{B}_i} \leq 2t \cdot \OPT$.
\end{lemm}
\begin{proof}[sketch]
	Along the lines of~\cite[page~19]{Raghvendra18}, we obtain in our setting
	$$\sum_{i=1}^{h} \dist(\Mon_{\mathcal{B}_i}) \leq O(\tfrac{1}{\eps}) \cdot \sum_{i=1}^{h} \OPT_{\mathcal{B}_i} \leq O(\tfrac{1}{\eps}) \cdot \OPT.$$
	In contrast to~\cite{Raghvendra18}, we do not consider all edges of~\Mon at the same time and require the use of Lemma~\ref{lem:ConsecutivePaths}.
	To this end, we need to argue that it is possible to assume that only requests in blocks below $B_1,B_2,\ldots,B_h$ appeared so far.
	Since the arrival of a request~$r_h$ only influences the matching on the~$\I_h$-portion of the line, we may alter the arrival order via any permutation~$\pi$ that satisfies:~$\I_i \subseteq \I_j$ and~$i < j$ imply~$\pi(i) < \pi(j)$.
	Thus, we can assume that requests above $B_1,B_2,\ldots,B_h$ arrive last.
	In fact, we may discard such requests altogether and consider the intermediate matching before their arrival, as all results for the final matching are true for intermediate ones.
	The additional additive factor which we incurred in Lemma~\ref{lem:ConsecutivePaths} is accounted for in the~$O(\tfrac{1}{\eps})$\nobreakdash-term.
	\qed
\end{proof}

Applying Lemma~\ref{lem:boundBlocks} to the top blocks of the recourse matching, implies Lemma~\ref{lem:activeEdges}.

\medskip 
\noindent \textbf{Bounding the recourse.}\label{sec:freezingScheme}
%
Freezing a request $r$ does not cause any reassignments.
Therefore, unfreezing it can cause reassignments (Step~3) at most once between two consecutive recourse steps that involve~$r$.
We charge both recourse actions for unfreezing~$r$ (described above, see Fig.~\ref{fig:unfreezing}) to~$r$ itself. It suffices to bound the number of recourse steps in which~$r$ can be involved.
\begin{obs}\label{obs:ChargeUnfreezingToRecourseSteps}
	Each request $r$ is reassigned at most three times per recourse step that involves~$r$.
\end{obs}

Recall, that for a request that is frozen at time~$i$ until it becomes unfrozen at time~$j$, we have~$\dist_e \leq \tfrac{\OPT_i}{i^2}$ and~$\dist_e > \tfrac{\OPT_j}{j}$.
Noting the fact that~$\OPT_i \leq \OPT_j$ whenever~$i \leq j$, this leads to the following easy observation.

\begin{obs} 
An edge frozen at time~$i$ stays frozen until time at least~$i^2$. \label{obs:frozen}
\end{obs}

On the other hand, the number of recourse steps in which a continuously non-frozen edge takes part in can be bounded from above.
\begin{lemm}\label{lem:freeze}
If a request~$r\in\I_i\subseteq\I_j$ is not frozen from time~$i$ to~$j$, then~$\level{\I_{j}}-\level{\I_{i}} = O(\log j)$. 
In particular, the number of recourse steps, that $r$ participates in between time~$i$ and time~$j$ is  bounded by $O(\log j)$.
\end{lemm}
\begin{proof}
Consider the search interval~\smash{$\SI_{j}=[s_j^\mathrm{L},s_j^\mathrm{R}]$}.
By definition, there exist augmenting paths~$P_j^\mathrm{L},P_j^\mathrm{R}$ from~$r_{j}$ to~$s_j^\mathrm{L}$ and~$s_j^\mathrm{R}$ 
with the same~$t$-net-cost as~$P_{j}$, the 
path used by the \talg. 
We bound the length of~$P_{j} \in \{P_j^\mathrm{L},P_j^\mathrm{R}\}$~by
\begin{equation}\label{eq:BoundAugmPath}
\dist(P_{j}) = \dist(P_{j} \cap \Mof_{j -1}) + \dist(P_{j} \cap \Mof_{j}) \leq 2t \cdot \OPT_j.
\end{equation}
Without loss of generality, assume~\smash{$P_{j} = P_{j}^\mathrm{L}$}. We 
could augment 
$\Mof_{j -1}$ also along~$P_{j}^\mathrm{R}$ obtaining a different matching~$\tilde{M}^*_{j}$.
By definition of $P_j^\mathrm{R}$ and~\eqref{eq:tnetcost},~$$t\cdot\dist(P_j \cap \Mof_{j}) - \dist(P_j \cap \Mof_{j-1})= \smash{\phi_t^{\Mof_{j -1}}(P_j) = \phi_t^{\Mof_{j -1}}(P_j^\mathrm{R})} = t\cdot\dist(P_j \cap \tilde\Mof_{j}) - \dist(P_j \cap \Mof_{j-1}).$$ Thus,
\begin{align*}
	\dist(P_{j}^\mathrm{R}) &= \dist(P_{j}^\mathrm{R} \cap \Mof_{j -1}) + \dist(P_{j}^\mathrm{R} \cap \tilde{M}^*_{j}) \\ &\leq \dist(P_{j}^\mathrm{R} \cap \Mof_{j -1}) + \dist(P_{j} \cap \Mof_{j}) + \tfrac{1}{t} \dist(P_{j}^\mathrm{R} \cap \Mof_{j -1}) 
	\leq 3t \cdot \OPT_j.
\end{align*}
With Equation~(\ref{eq:BoundAugmPath}) we get~$\abs{\SI_{{j}}} \leq \dist( P_j^\mathrm{L}) + \dist( P_j^\mathrm{R}) \leq 5t \cdot \OPT_j,$
which implies $$\smash{\abs{\I_{{j}}} \leq \bigcup_{k\leq j}\abs{\SI_k} \leq 5t \cdot j \cdot \OPT_{j}.}$$
	
On the other hand, the cost of~$\abs{\I_{i}}$ can be bounded from below by~$\dist_e$, with~$e$ being the edge in~$\Mon$ incident to~$r$, as both ends of~$e$ are contained in the interval.
We then obtain
	$$\frac{\abs{\I_{j}}}{\abs{\I_{i}}} \leq \frac{5t \cdot j \cdot \OPT_{j}}{\dist_e} < 5t \cdot j^3.$$
The last inequality follows from the assumption that~$r$ is not frozen at time~$j$, whereby~$\dist_e > \smash{\tfrac{\OPT_{j}}{j^2}}$.
Recall that $\abs{\I_i} < (1+\eps)^{\level{\I_i}}$ and $(1+\eps)^{\level{\I_j}-1} \leq \abs{\I_j}$.
We conclude
\begin{equation}
	\level{\I_{j}}-\level{\I_{i}} \leq \log_{(1+\eps)}(5 t \cdot j^3) +1 = c\cdot \log j \label{eq:c}\,,
\end{equation}
for some constant $c$. Regarding the second claim, recall that $r$ participates in a recourse step, whenever an interval containing~$r$ opens a new top block (i.e. a new level) while~$r$ is not frozen. Between time~$i$ and time~$j$, this can only happen for intervals $\I_h$ of distinct levels for which $\I_i\subseteq\I_h\subseteq\I_j$. The claim follows.
	\qed
\end{proof}

Lemma~\ref{lem:freeze} together with Observation~\ref{obs:frozen} allows us to bound the total number of recourse actions expended by Algorithm~\ref{alg:general}.
\begin{lemm}\label{lem:recourse}
Algorithm~\ref{alg:general} uses a recourse budget of at most~$O(n \log n)$.
\end{lemm}
\begin{proof}
 	Consider a request~$r$.
	By Observation~\ref{obs:ChargeUnfreezingToRecourseSteps}, it suffices to bound the number of recourse steps that $r$ is involved in.
	Let~$[i^U_h,i^F_h]$,~$h=0,1,\ldots ,k$, be maximal intervals of consecutive time points during which $r$ is not frozen, i.e., $r$ is not frozen at any time $i\in[i^U_h,i^F_h]$. 	
	We use induction on~$k$ to show that~$r$ participates in at most~$2c\cdot \log (i^F_{k})$ recourse steps, where~$c$ is the constant from Equation~\eqref{eq:c}.
	The base case,~$k=0$, follows directly from Lemma~\ref{lem:freeze}.	
	For~$k\geq 1$, we have~$(i^F_{k-1})^2 \leq i^U_{k} \leq i^F_{k}$ due to Observation~\ref{obs:frozen}.
	By induction hypothesis, the number of reassignments that involve $r$ in the first~$k-1$ time intervals is at most
	$$ 2c \cdot \log (i^F_{k-1}) \leq 2c \cdot \log \!\Big( \sqrt{i^F_{k}} \; \Big) = c \cdot \log (i^F_{k}).$$
	For the last time interval, we have at most~$c \cdot \log (i^F_{k})$ many such recourse steps by Lemma~\ref{lem:freeze}.
	Since~$i^F_{k} \leq n$, this concludes the proof.
\qed
\end{proof}

Corollary~\ref{cor:comp} and Lemma~\ref{lem:recourse} together imply Theorem~\ref{thm:general}.

\section{Near-Optimality on Alternating Instances}\label{sec:alternatingALG}

For alternating instances, we may assume that requests and servers alternate 
from $-\infty$ to~$\infty$ on the line,  with servers at~$\pm \infty$.
For such instances, an optimum matching matches all requests either to the server directly to their left or all to the right.
Denote these matchings by~$M^\mathrm{L}$ and~$M^\mathrm{R}$ 
and call their edges \emph{minimal}.

We describe a $(1+\eps)$-competitive algorithm for alternating instances that reassigns a request a constant number of times.
In addition to its output~$M$, it maintains
~\Mof and frozen edges~\Mf.
A request is 
frozen when it is reassigned the $k$-th time, for some 
$k$ only depending on \eps.
A frozen request remains matched to its current server and the corresponding edge of~\Mof is added to~\Mf.
Non-frozen requests are matched according to the detour matching~$\Mof \sym \Mf$. 
By design, the recourse budget per request is constant, only the competitive analysis remains.

\smallskip
\noindent
\emph{Notation.} We use a similar interval structure as before and keep the notation. 
Consider intervals~$\I_i = [s_i^\mathrm{L},s_i^\mathrm{R}]$, where~$s_i^\mathrm{L}, s_i^\mathrm{R} \in S$ are the closest free servers on the line to the left and right of~$r_i$ respectively.
Denote by~$P_i^\mathrm{L},P_i^\mathrm{R}$ the alternating paths connecting~$r$ to~$s_i^\mathrm{L}$ and~$s_i^\mathrm{R}$ respectively that have shortest~$t$-net-cost.

When considered as line segments, the augmenting paths~$P_i$ have a laminar structure.
View them as nodes of a forest, where~$P_i$ is a child of the minimal augmenting path that properly contains it, or a root if no such path exists (see Figure~\ref{fig:pathtree}).
The \emph{depth} of a path~$P_i$ is its distance to the root and it determines the number of reassignments of the corresponding request~$r_i$ in~\Mof.

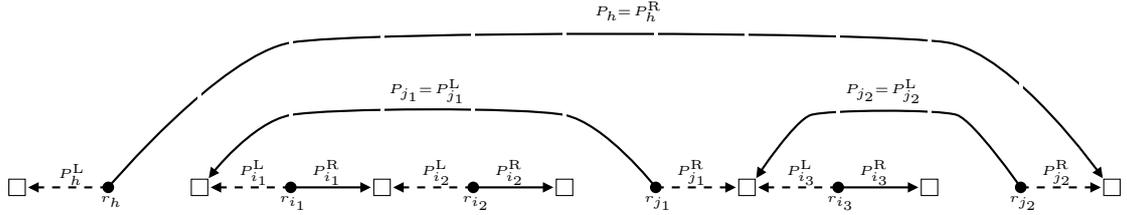
\begin{figure}[tbh]\centering
	\begin{tikzpicture}[scale=1.2,>={Stealth[inset=0pt,length=4pt,angle'=50]}]	
	\foreach \r / \s in {3,5,9}{	
		\draw[augment, ->](\r,0) -- (\r+.88,0);
		\draw[alternative] (\r,0) -- (\r-.88,0);
	}	
	\draw[alternative] (1,0) -- (0.12,0);
	\draw[alternative] (7,0) -- (7.88,0);
	\draw[alternative] (11,0) -- (11.88,0);
	\draw[augment, ->] plot [smooth, tension = .4] coordinates { (7,0) (6,.8) (3,.8) (2.1,0.1)};
	\draw[augment, ->] plot [smooth, tension = .3] coordinates { (11,0) (10.3,.8) (8.7,.8) (8.1,0.1)};
	\draw[augment, ->] plot [smooth, tension = .3] coordinates { (1,0) (3,1.6) (10.2,1.6) (11.9,0.1)};
	\foreach \x  in {2,3,4,5,6,7,8,9,10,11}{	
		\fill[white] (\x-.02,-.3) rectangle (\x+.02,2);
	}	
	\foreach \x / \label in {0/h,1/i_1,2/i_2,3/j_1,4/i_3,5/j_2}{	
		\node[server]() at (2*\x,0) {};
		\node[request]() at (2*\x+1,0) {};
		\node[below]() at (2*\x+1,0) {\tiny~$r_{\label}$};
	}		
	\node[server]() at (12,0) {};
	\foreach \x / \y / \label in {2.6/0/{P_{i_1}^\mathrm{L}},3.4/0/{P_{i_1}^\mathrm{R}},4.6/0/{P_{i_2}^\mathrm{L}},5.4/0/{P_{i_2}^\mathrm{R}},8.6/0/{P_{i_3}^\mathrm{L}},9.4/0/{P_{i_3}^\mathrm{R}},4.5/.9/{P_{j_1}\!\! = \!P_{j_1}^\mathrm{L}},7.4/0/{P_{j_1}^\mathrm{R}},9.5/.9/{P_{j_2}\!\! = \!P_{j_2}^\mathrm{L}},11.4/0/{P_{j_2}^\mathrm{R}},.6/0/{P_{h}^\mathrm{L}},6.7/1.8/{P_{h}\!\! = \!P_{h}^\mathrm{R}}}{	
		\node[above = -.08]() at (\x,\y) {\tiny$\label$};
	}
	\end{tikzpicture}
	\caption{Illustration of a path tree in an alternating instance. Paths not chosen for augmentation are dashed. Servers are depicted as squares and requests as filled circles.}\label{fig:pathtree}
\end{figure}


\begin{lemm}\label{lem:alternating}
\begin{enumerate}
	\item[(i)] Paths~$P_i^\mathrm{L},P_i^\mathrm{R}$ and matching~$\Mof_{i}$ only use minimal edges.
	\item[(ii)] If~$P_i = P_i^{X}\!\!$, for~$X\in\{\mathrm{L},\mathrm{R}\}$, then in the area of~$\I_i$, locally, we have~$\Mof_i = M^\mathrm{X}$.
Specifically, this implies~$\phi_t^{i-1}(P_i) = \phi_t^{\vphantom{i-1}\smash{M^Y}}\!\!(P_i)$ for~$X \neq Y\in\{\mathrm{L},\mathrm{R}\}$.	
	\item[(iii)] If~$P_{j}$ is a child of~$P_{i}$, then~$P_i = P_i^\mathrm{L}$ if and only if~$P_j = P_j^\mathrm{R}$. 
In particular, we have~$\I_j \subseteq P_i$ and~$\I_{j} \cap \Mof_j \cap \Mof_i = \emptyset$.

\end{enumerate}

\end{lemm}
\begin{proof}
	(i): We use induction on~$i$.
	The base case,~$i=1$, is easy.
	Let~$i\geq 2$.
	Without loss of generality~$P_i = P_i^\mathrm{R}$. Assume $P_i$ contains a non-minimal edge~$e$.  
	Let the notation $(s,r)$ for an edge reflect that $s$ is to the left of $r$ on the line.
	Edges~$(s,r)\in P_i$ are contained in~$\Mof_{i-1}$ and by induction hypothesis minimal, so~$e=(r,s)\in\Mof_{i}$.
	Consider the server~$s'$ right of~$r$.
	Since~$e$ is not minimal,~$r < s' < r' < s$, with~$r'$ being the server right of~$s'$.
	If~$s'$ was free, altering $P_i$ to go from~$r$ directly to~$s'$ would yield a lower cost.
	This follows from the fact that the~$t$-net-cost of a path from a request to a free server is always non-negative~\cite{Raghvendra16}.
	If~$s'$ is matched, it must be matched to~$r'$. Therefore, replacing~$e$ by~$(r,s'),(s',r'),(r',s)$ reduces the cost as well, contradict minimality of $P_i$.
	
	(ii): By statement (i),~$P_i^\mathrm{L}$ and~$P_i^\mathrm{R}$ only consist of minimal edges.
	Then~$P_i^\mathrm{L} \cap \Mof_{i-1} = M^\mathrm{R}$ and~$\Mof_{i-1} \cap P_i^\mathrm{R} = M^\mathrm{L}$.
	After augmenting~$\Mof_{i-1}$ along~$P_i$, the matching is flipped.
	
	(iii): By (ii), we know that~\smash{$\Mof_j \cap \I_j = M^\mathrm{X}$}, say~$X=R$, so edges are of the form~$(r,s)$.
	From part (i), only augmenting paths~$P_h=P_h^\mathrm{L}$ can traverse them.
	If this happens, all of~$\I_j$ is traversed as there is no free server in its interior.
	As parent of $P_j$,~$\,P_i$ is the first path to properly contain~$P_j$ and thus~$P_i=P_i^\mathrm{L}$ and in particular~$\I_j \subseteq P_i$.
	The equation~$\I_{j} \cap \Mof_j \cap \Mof_i = \emptyset$ follows directly from (ii).	
	\qed 
\end{proof}
We show that the sum of lengths of augmenting paths of some depth grows exponentially towards the root.
For a path $P_h$, let $|P_h|$ denote the length of the corresponding line segment. Further, let $\mathcal H_k$ be the set of indices of paths at depth $k$ in the induced subtree of $P_h$ with the root $P_h$ at depth $0$.
\begin{lemm}\label{lem:IntervalForest}
Consider a path~$P_h$ and its grandchildren~$P_j$,~$j \in  \mathcal H_2$.
Then
$$ |P_h| \geq \left(2-\tfrac{1}{t}\right)\sum_{j\in \mathcal H_2} |P_j|.$$
\end{lemm}
\begin{proof}
Denote by~$P_i$,~$i \in  \mathcal{I}$, the children of~$P_h$ in the path-forest and by~$\mathcal J_i \subseteq \mathcal{J}$ the sets of indices of their respective children.
Without loss of generality~\mbox{$P_h = P_h^\mathrm{R}$}. Lemma~\ref{lem:alternating}, (iii) implies~$P_i = P_i^\mathrm{L}$, and~$P_j = P_j^\mathrm{R}$, for~\mbox{$i \in \mathcal{I},j \in \mathcal{J}$.}
See Figure~\ref{fig:pathtree} for an illustration.
With again Lemma~\ref{lem:alternating}, (iii) and~\mbox{$\phi_t(P_i^\mathrm{R}) \geq \phi_t(P_i^\mathrm{L})$}, we get
\begin{align*}
t \cdot |P_h| &\geq t \cdot  \sum_{i\in\mathcal{I}}|I_i| = t \cdot  \sum_{i\in\mathcal{I}} \big( |P_i^\mathrm{L}| + |P_i^\mathrm{R}| \big) \geq  \sum_{i\in\mathcal{I}} \big(  t \cdot |P_i^\mathrm{L}| + \phi_t^{i-1}(P_i^\mathrm{R}) \big) \\
&\geq t \cdot \sum_{j\in\mathcal{J}}  |P_j^\mathrm{R}| + t \cdot \sum_{i\in\mathcal{I}} |P_i \setminus (\cup_{j \in \mathcal{J}_i} P_j^\mathrm{R})| + \sum_{i\in\mathcal{I}} \phi_t^{i-1}(P_i^\mathrm{L}).
\end{align*}
Using
\begin{align*}
\phi_t^{i-1}(P_i^\mathrm{L}) = \phi_t^{\mathrm{R}}(P_i^\mathrm{L}) &= \sum_{j\in\mathcal{J}_i} \!\!\Big( \phi_t^{\mathrm{R}}(P_j^\mathrm{L}) + \phi_t^{\mathrm{R}}(P_j^\mathrm{R}) \Big) + \phi_t^{\mathrm{R}}(P_i \setminus (\cup_{j \in \mathcal{J}_i} \I_j))\\
&\geq \sum_{j\in\mathcal{J}_i} \!\Big( \phi_t^{j-1}(P_j^\mathrm{L}) + \psi_t^{j-1}(P_j^\mathrm{R}) \Big) - t\cdot | P_i \setminus (\cup_{j \in \mathcal{J}_i} \I_j) |,
\end{align*}
and~$\phi_t(P_j^\mathrm{L}) \geq \phi_t(P_j^\mathrm{R})$, we obtain	
\begin{align*}
t \cdot |P_h| 
&\geq \sum_{j\in\mathcal{J}} \!\Big( t \cdot|P_j^\mathrm{R}| + \phi_t^{j-1}(P_j^\mathrm{R}) + \psi_t^{j-1}(P_j^\mathrm{R}) \Big) \geq (2t -1) \sum_{j\in\mathcal{J}} \cdot|P_j^\mathrm{R}|.
\end{align*}
The last inequality follows from our observation \mbox{$\psi_t(P) + \phi_t(P) = (t-1)\cdot|P|$}.\\
\qed
\end{proof}

\begin{proof}[of Theorem~\ref{thm:alternating}]
We in fact prove a stronger result than in the theorem statement and show that the algorithm described in the beginning of this section is~$(1+\eps)$-competitive while reassigning each request at most $O(\eps^{-(1 + \lambda)})$ times for fixed $\lambda >0$. Consider a path~$P_{h}$.
The intervals~$\I_j$,~$j \in \mathcal H_{2k+2}$, are contained in paths~$P_{j'}$,~$j' \in \mathcal H_{2k+1}$, by Lemma~\ref{lem:alternating} (iii). Lemma~\ref{lem:IntervalForest} implies
\begin{align}
\textstyle \sum_{j\in \mathcal H_{2k+2}} |I_j| \leq \sum_{j'\in \mathcal H_{2k+1}} |P_{j'}| \leq \left(2-\tfrac{1}{t}\right)^{-k} \cdot \sum_{i\in \mathcal H_{1}} |P_{i}|. \label{eq:2k+1Children}
\end{align} 

Raghvendra~\cite{Raghvendra16} showed that the~$t$-net-cost of augmenting paths is always non-negative. 
In particular,~$t \cdot \dist(P_h \cap \Mof_{h}) -  \dist(P_h \cap \Mof_{h-1}) \geq 0$ and, thus,
\begin{align}
|P_{h}| =  \dist(P_h \cap \Mof_{h}) +  \dist(P_h \cap \Mof_{h-1}) \leq (t+1) \cdot \dist(P_h \cap \Mof_{h}). \label{eq:pathCapMof}
\end{align}

Similarly,	
$ \sum_{i\in I} |P_{i}| \leq (t+1) \cdot \sum_{i\in I} \dist(P_i \cap \Mof_{i}) \leq (t+1) \cdot \dist(P_h \cap \Mof_{i})$.
In an interval~$\I_i$, locally,~$\Mof_h=M^{\mathrm{L}}$ if and only if~\mbox{$\Mof_i=M^{\mathrm{R}}$} by Lemma~\ref{lem:alternating} (ii). The above together with~(\ref{eq:2k+1Children}) implies
\begin{align}
\left[\tfrac{1}{t+1}\left(2-\tfrac{1}{t}\right)^{k} -1\right]\cdot \sum_{j\in J} |\I_j| \leq \dist\left(P_h \setminus (\cup_{j \in \mathcal{J}_i} \I_j) \cap M^{\mathrm{X}} \right), \quad X \in \{\mathrm{L},\mathrm{R}\}. \label{eq:chargealternating}
\end{align}

Denote by~$\alpha(k,t)$ the term in square brackets.
When a (minimal) edge 
is frozen, the remaining instance 
is again alternating and
at most one request will 
take a detour due to~$(r,s)$ being frozen.
The additional 
cost is botunded by~$|\I_r|$ and can, via Equation~(\ref{eq:chargealternating}), be charged to non-frozen parts of \Mopt. 
A part $P_h\cap\Mopt$ is charged this way at most~$2k+2$ times before $P_h$ itself is frozen which leads to a competitive factor of~$\smash{(t+\tfrac{2k+2}{\alpha(k,t)})}$.
Setting~$t = 1 + \smash{\tfrac{\eps}{2}}$, and substituting $\alpha(k,t)$ from Equation~(\ref{eq:chargealternating}), we see that this term is at most $1+\eps$ if
$$
	 \tfrac{2}{4+\eps}\big(1+\tfrac{\eps}{2+\eps}\big)^k - \tfrac{\eps}{2} \geq 2k+2.  
$$
Fix some $\lambda > 0$. Assuming $\eps < 1$, and setting $k=\frac{4c}{\eps} \cdot \frac{\eps^{-\lambda}}{\lambda}$ for some constant $c$, this simplifies to
$$
 \tfrac{2}{5}\bigg(\big(1+\tfrac{\eps}{3}\big)^{\frac{4}{\eps}}\bigg)^{c \cdot \frac{\eps^{-\lambda}}{\lambda}} - \tfrac{1}{2} \geq 	2\left( \tfrac{4c}{\eps} \cdot \tfrac{\eps^{-\lambda}}{\lambda}\right)+2 .  
$$ 
Using the fact that $\big(1+\tfrac{\eps}{3}\big)^{\frac{4}{\eps}}\geq e$ for $\eps < 1$ and simplifying further, we obtain
\begin{align}\label{eq:altComp}
e^{\frac{c \cdot \eps^{-\lambda}}{\lambda}} \geq \frac{20c}{\eps} \cdot \frac{\eps^{-\lambda}}{\lambda} +7.
\end{align}

Since the left term is exponential in $\eps^{-\lambda}$, it dominates the polynomial right term and the inequality holds for $\eps$ sufficiently small (dependent on $\lambda$). To remove the dependence of \eps on $\lambda$, we may instead choose $c = c(\lambda)$ sufficiently large. This yields a recourse of $O_\lambda(\eps^{-(1+\lambda)})$ as desired.
 \qed
\end{proof}

As a byproduct, we show for this special class of instances a result in the online setting without recourse. It relates the competitive ratio to the cost metric, i.e., the maximum difference in edge cost for connecting a request to a server. This result compares to the best known competitive ratio of $O(\log n)$ by Raghvendra~\cite{Raghvendra18}.
\begin{theo}\label{thm:alternatingO}
The online \talg\ is $O(\log \Delta)$-competitive for online matching on an alternating line, where~\smash{$\Delta= \max_{r,r'\in R\!,s,s'\in S}\tfrac{\dist(r,s)}{\dist(r'\!,s')}$}.
\end{theo}
\begin{proof}
	Consider an edge~$e_i$ of \Mon and assume~$P_{i} = P_{i}^\mathrm{R}$. 
	By Lemma~\ref{lem:alternating}, (ii), \vspace{-.5 em}
	\begin{align*}
	&t \cdot \dist( P_{i}^\mathrm{L} \cap M^\mathrm{L} ) - \dist( P_{i}^\mathrm{L} \cap M^\mathrm{R} ) = \phi_t^{i-1}(P_{i}^\mathrm{L}) \\
	&\geq \phi_t^{i-1}(P_{i})  = t \cdot \dist( P_{i} \cap M^\mathrm{R} ) - \dist( P_{i} \cap M^\mathrm{L} ).
	\end{align*}
	Therefore,
	$$|P_{i}| = \dist( P_{i} \cap M^\mathrm{L} ) + \dist( P_{i} \cap M^\mathrm{R} ) \leq \left(1+\tfrac{1}{t}\right)\cdot\dist( P_{i} \cap M^\mathrm{L} ) + \dist( P_{i}^\mathrm{L} \cap M^\mathrm{L}).$$
	As in Equation~\ref{eq:pathCapMof}, we obtain 
	$|P_{i}| \leq (t+1) \cdot \dist(P_i \cap M^\mathrm{R})$. Together with the above, it implies that
	$\dist_{e_i} = |P_{i}| \leq \dist(\I_i \cap \Mopt) \cdot \max\{ t+1, 1+\tfrac{1}{t} \}$.
	Since intervals corresponding to the same depth in the path-forest are disjoint, we can bound the lengths of paths which are of same depth by~$\max\{ t+1, 1+\tfrac{1}{t} \}\cdot \OPT$. As there are at most~$2\cdot\log_{\left(2-1/t\right)}\Delta$ levels in total, the theorem's statement follows.
	\qed
\end{proof}

\section{Conclusion} In this paper, we give the first non-trivial results for the min-cost online bipartite matching problem with recourse. (The results were obtained simultaneously with and independently of Gupta et al.~\cite{GuptaKS19} who consider also more general metrics than the line.) We confirm that an average recourse of $O(\log n)$ per request is sufficient to obtain an $O(1)$-competitive matching on the line. It remains open if such a result can be obtained in a non-amortized setting, where the recourse is available only per iteration. Our algorithm is clearly designed for the amortized setting as it buffers online matching decisions and repairs them in batches.

Further, it remains open to show that constant recourse per request is sufficient for maintaining an $O(1)$-competitive matching on the line, as it is the case for the special line with alternating requests. (In our notion, this would be a total number of reassignments of $O(n)$.) This may be very well possible as there is, currently, no lower bound that would rule this out.

Finally, we remind of a major open question in this field: Does there exist an $O(1)$-competitive algorithm for online matching on the line without any recourse?

\bibliographystyle{abbrv}
\bibliography{mybibliography}

\begin{thebibliography}{10}

\bibitem{AngelopoulosDJ18}
S.~Angelopoulos, C.~D{\"{u}}rr, and S.~Jin.
\newblock Online maximum matching with recourse.
\newblock In {\em {MFCS}}, volume 117 of {\em LIPIcs}, pages 8:1--8:15. Schloss
  Dagstuhl - Leibniz-Zentrum fuer Informatik, 2018.

\bibitem{AntoniadisBNPS19}
A.~Antoniadis, N.~Barcelo, M.~Nugent, K.~Pruhs, and M.~Scquizzato.
\newblock A o(n)-competitive deterministic algorithm for online matching on a
  line.
\newblock {\em Algorithmica}, 81(7):2917--2933, 2019.

\bibitem{DBLP:conf/latin/AntoniadisFT18}
A.~Antoniadis, C.~Fischer, and A.~T{\"{o}}nnis.
\newblock A collection of lower bounds for online matching on the line.
\newblock In {\em {LATIN}}, volume 10807 of {\em Lecture Notes in Computer
  Science}, pages 52--65. Springer, 2018.

\bibitem{BansalBSN2014}
N.~Bansal, N.~Buchbinder, A.~Gupta, and J.~S. Naor.
\newblock A randomized {O}$(log^2 k)$-competitive algorithm for metric
  bipartite matching.
\newblock {\em Algorithmica}, 68(2):390--403, 2014.

\bibitem{BernsteinHR19}
A.~Bernstein, J.~Holm, and E.~Rotenberg.
\newblock Online bipartite matching with amortized \emph{O}(log
  \({}^{\mbox{2}}\) \emph{n}) replacements.
\newblock {\em J. {ACM}}, 66(5):37:1--37:23, 2019.

\bibitem{BosekLSZ14}
B.~Bosek, D.~Leniowski, P.~Sankowski, and A.~Zych.
\newblock Online bipartite matching in offline time.
\newblock In {\em Proceedings of {FOCS}}, pages 384--393. {IEEE} Computer
  Society, 2014.

\bibitem{ChaudhuriDKL09}
K.~Chaudhuri, C.~Daskalakis, R.~D. Kleinberg, and H.~Lin.
\newblock Online bipartite perfect matching with augmentations.
\newblock In {\em Proceedings of {INFOCOM}}, pages 1044--1052. {IEEE}, 2009.

\bibitem{DuanP14}
R.~Duan and S.~Pettie.
\newblock Linear-time approximation for maximum weight matching.
\newblock {\em J. ACM}, 61(1):1:1--1:23, 2014.

\bibitem{DBLP:journals/tcs/FuchsHK05}
B.~Fuchs, W.~Hochst{\"{a}}ttler, and W.~Kern.
\newblock Online matching on a line.
\newblock {\em Theor. Comput. Sci.}, 332(1-3):251--264, 2005.

\bibitem{GairingK19}
M.~Gairing and M.~Klimm.
\newblock Greedy metric minimum online matchings with random arrivals.
\newblock {\em Oper. Res. Lett.}, 47(2):88--91, 2019.

\bibitem{GroveKKV95}
E.~F. Grove, M.~Kao, P.~Krishnan, and J.~S. Vitter.
\newblock Online perfect matching and mobile computing.
\newblock In {\em Proceedings of {WADS}}, volume 955 of {\em LNCS}, pages
  194--205. Springer, 1995.

\bibitem{GuG016}
A.~Gu, A.~Gupta, and A.~Kumar.
\newblock The power of deferral: Maintaining a constant-competitive steiner
  tree online.
\newblock {\em {SIAM} J. Comput.}, 45(1):1--28, 2016.

\bibitem{GuptaGPW19}
A.~Gupta, G.~Guruganesh, B.~Peng, and D.~Wajc.
\newblock Stochastic online metric matching.
\newblock In {\em {ICALP}}, volume 132 of {\em LIPIcs}, pages 67:1--67:14.
  Schloss Dagstuhl - Leibniz-Zentrum fuer Informatik, 2019.

\bibitem{GuptaKS14}
A.~Gupta, A.~Kumar, and C.~Stein.
\newblock Maintaining assignments online: Matching, scheduling, and flows.
\newblock In {\em {SODA}}, pages 468--479. {SIAM}, 2014.

\bibitem{GuptaL12}
A.~Gupta and K.~Lewi.
\newblock The online metric matching problem for doubling metrics.
\newblock In {\em {ICALP}}, volume 7391 of {\em LNCS}, pages 424--435.
  Springer, 2012.

\bibitem{GuptaKS19}
V.~Gupta, R.~Krishnaswamy, and S.~Sandeep.
\newblock Permutation strikes back: The power of recourse in online metric
  matching, 2019.

\bibitem{ImaseW91}
M.~Imase and B.~M. Waxman.
\newblock Dynamic steiner tree problem.
\newblock {\em {SIAM} J. Discrete Math.}, 4(3):369--384, 1991.

\bibitem{KalyanasundaramP93}
B.~Kalyanasundaram and K.~Pruhs.
\newblock Online weighted matching.
\newblock {\em J. Algorithms}, 14(3):478--488, 1993.

\bibitem{KarpVV90}
R.~M. Karp, U.~V. Vazirani, and V.~V. Vazirani.
\newblock An optimal algorithm for on-line bipartite matching.
\newblock In {\em {STOC}}, pages 352--358. {ACM}, 1990.

\bibitem{KhullerMV94}
S.~Khuller, S.~G. Mitchell, and V.~V. Vazirani.
\newblock On-line algorithms for weighted bipartite matching and stable
  marriages.
\newblock {\em Theor. Comput. Sci.}, 127(2):255--267, 1994.

\bibitem{Kuhn55thehungarian}
H.~W. Kuhn and B.~Yaw.
\newblock The hungarian method for the assignment problem.
\newblock {\em Naval Res. Logist. Quart}, pages 83--97, 1955.

\bibitem{LackiOPS2015}
J.~Lacki, J.~O\'{c}wieja, M.~Pilipczuk, P.~Sankowski, and A.~Zych.
\newblock The power of dynamic distance oracles: Efficient dynamic algorithms
  for the steiner tree.
\newblock In {\em Proceedings of the {STOC}}, pages 11--20. ACM, 2015.

\bibitem{MatuschkeSV19}
J.~Matuschke, U.~Schmidt{-}Kraepelin, and J.~Verschae.
\newblock Maintaining perfect matchings at low cost.
\newblock In {\em {ICALP}}, volume 132 of {\em LIPIcs}, pages 82:1--82:14.
  Schloss Dagstuhl - Leibniz-Zentrum f{\"u}r Informatik, 2019.

\bibitem{MegowSVW16}
N.~Megow, M.~Skutella, J.~Verschae, and A.~Wiese.
\newblock The power of recourse for online {MST} and {TSP}.
\newblock {\em {SIAM} J. Comput.}, 45(3):859--880, 2016.

\bibitem{Mehta13}
A.~Mehta.
\newblock Online matching and ad allocation.
\newblock {\em Foundations and Trends® in Theoretical Computer Science},
  8(4):265--368, 2013.

\bibitem{NayyarR17}
K.~Nayyar and S.~Raghvendra.
\newblock An input sensitive online algorithm for the metric bipartite matching
  problem.
\newblock In {\em {FOCS}}, pages 505--515. {IEEE} Computer Society, 2017.

\bibitem{Raghvendra16}
S.~Raghvendra.
\newblock A robust and optimal online algorithm for minimum metric bipartite
  matching.
\newblock In {\em {APPROX-RANDOM}}, volume~60 of {\em LIPIcs}, pages
  18:1--18:16. Schloss Dagstuhl - Leibniz-Zentrum fuer Informatik, 2016.

\bibitem{Raghvendra18}
S.~Raghvendra.
\newblock Optimal analysis of an online algorithm for the bipartite matching
  problem on a line.
\newblock In {\em Symposium on Computational Geometry}, volume~99 of {\em
  LIPIcs}, pages 67:1--67:14. Schloss Dagstuhl - Leibniz-Zentrum fuer
  Informatik, 2018.

\end{thebibliography}

\end{document}